\documentclass[prl,twocolumn]{revtex4-1}

\usepackage{amsmath}
\usepackage{amssymb}
\usepackage{amsthm}
\usepackage{color}
\usepackage{mathrsfs}
\usepackage{graphicx}
\newtheorem{proposition}{Proposition}
\newtheorem{lemma}{Lemma}
\newtheorem{theorem}{Theorem}

\def\openone{\leavevmode\hbox{\small1 \normalsize \kern-.64em1}}
\def\identity{\leavevmode\hbox{\small1\kern-3.8pt\normalsize1}}
\newcommand{\E}{\ensuremath \mathbb E}

\renewcommand{\Pr}{\ensuremath \mathrm {Pr}}

\newcommand{\Tr}{\operatorname{Tr}}

\begin{document}
\title{Lower Bounds on the Communication Complexity of Binary Local Quantum Measurement Simulation}
\author{Adrian Kosowski}
\affiliation{Inria Bordeaux Sud-Ouest --- LaBRI, 351 cours de la Lib\'eration, 33405 Talence Cedex, France}
\author{Marcin Markiewicz}
\affiliation{Institute of Theoretical Physics and Astrophysics, University of Gda\'nsk, 80-952 Gda\'nsk, Poland}

\begin{abstract}
We consider the problem of the classical simulation of quantum measurements in the scenario of communication complexity. Regev and Toner (2007) have presented a 2-bit protocol which simulates one particular correlation function arising from binary projective quantum measurements on arbitrary state, and in particular does not preserve local averages. The question of simulating other correlation functions using a protocol with bounded communication, or preserving local averages, has been posed as an open one. Within this paper we resolve it in the negative: we show that any such protocol must have unbounded communication for some subset of executions. In particular, we show that for any protocol, there exist inputs for which the random variable describing the number of communicated bits has arbitrarily large variance.
\end{abstract}																														

\maketitle

\emph{Introduction.} It is known since Bell's seminal paper \cite{B64}, that quantum mechanics of many parties predicts correlations of measurement outcomes that cannot be reconstructed using only locally computable functions without any communication between the parties. This property is commonly referred to as \emph{quantum nonlocality} (not to be confused with nonlocality in the sense of special relativity, cf.~\cite{G09} for a comprehensive discussion of the topic). The degree of quantum nonlocality in a computational setting can be intuitively described using the scenario of communication complexity, in which one asks how much communication is needed to perform some distributed task. From this perspective, the communication cost of classically simulating quantum measurements is one of the most important approaches to characterizing the fundamental difference between classical and quantum correlations \cite{BCMW10}.

All further considerations are performed using the scenario of \emph{randomized communication complexity}  proposed by Yao \cite{Y79}. In a bipartite version of this scenario, two separate processors with no bounds on local memory and computational efficiency are considered. They can share random numbers (a so-called public coin), and can exchange bits with each other. They are given inputs (which in case of Bell inequalities scenario correspond to measurement settings), and their task is to produce discrete outputs. The statistical distribution of the outputs with respect to used random resources is required to fulfill some task distribution and the algorithm should be performed with minimal possible amount of communicated bits. 

Within the scenario of randomized bipartite communication complexity we consider the problem of simulating local quantum measurements, in which Alice and Bob are given as an input the classical description of their own local observable and of the entire quantum state, and their task is to output some numbers, whose statistical distribution is to some extent consistent with the distribution of measurement outcomes predicted by quantum theory. The problem can be treated in two different variants. In the first one ($CS$ -- Correlation Simulation) the task is to perfectly simulate correlations of outcomes of all possible quantum measurements, whereas no restrictions are imposed on local averages of outputs. The second approach ($LQMS$ -- Local Quantum Measurement Simulation) demands both the perfect simulation of correlations and the perfect simulation of local quantum averages.

The communication complexity of the above-defined problems is not known in general, and only partial results have been proposed. In the case of the first approach ($CS$), the problem has been solved in general for binary (two-outcome) projective measurements on an arbitrary quantum state by Regev and Toner \cite{RT07}, who proposed a bounded-worst-case randomized protocol which solves the problem with two bits of one-way communication. This result saturates the lower bound derived in \cite{VB09}, who showed that one bit is not sufficient. In the case of $d$-outcome projective measurements, a lower bound of $\Omega(d)$ bits of communication is known \cite{BCT99}, whereas for the special case of arbitrary POVM measurements on $n$ Bell states shared between Alice and Bob a protocol exists which uses $O(n 2^n)$ bits of communication on average \cite{MBCC01}. In the case of $LQMS$ only the case of a maximally entangled state of two qubits has been solved. Toner and Bacon \cite{TB03} have shown that in case of projective measurements on such a state, there exists a bounded-worst-case one-bit protocol, whereas in case of aribtrary POVM measurements there exists a protocol which uses six bits in expectation \cite{DLR07}.

In this paper we focus on the $LQMS$ problem for binary (two-outcome) observables, which we refer to as $BLQMS$. We construct a lower bound on all the statistical moments of the random variable describing the number of communicated bits in any protocol which exactly performs $BLQMS$. This lower bound implies that for any protocol solving $BLQMS$, the number of communicated bits has unbounded variance (as well as all higher moments) with respect to the dimension of the state space. This implies not only that there is no bounded-worst-case protocol for $BLQMS$ which works for arbitrary dimension of the state space, but also that for any $BLQMS$ protocol working in finite time the expected number of bits has a large dispersion. To the best of our knowledge, the presented characterization of communication complexity of $BLQMS$ is the most general among all already proposed, and gives an important insight into the characterization of \emph{quantum nonlocality} from the perspective of communication complexity.

\emph{Basic concepts and methods.} We formally define the Binary Local Quantum Measurement Simulation ($BLQMS$) problem as follows. The inputs of Alice and Bob are the pairs $(A,\sigma)$ and $(B, \sigma)$, where $\sigma$ is a density matrix of dimension $n^2 \times n^2$, representing an arbitrary state of two $n$-dimensional quantum systems, and $A$ and $B$ are Hermitian operators with spectrum $\{-1,1\}$. Alice and Bob are equipped with an infinite source of shared randomness. The protocol is randomized and the task of Alice and Bob is to output values $y_A, y_B \in \{-1,1\}$, respectively, such that the corresponding random variables $\mathbf{y}_A$, $\mathbf{y}_B$, taken over multiple runs of the protocol, have correlations and local averages statistically indistinguishable from those resulting from a quantum measurement of state $\sigma$ with the measurement operator $A \otimes B$, i.e.:
\begin{eqnarray}
&& \E (\mathbf{y}_A \mathbf{y}_B) = \Tr ((A\otimes B)\sigma),\nonumber\\
&& \E (\mathbf{y}_A) = \Tr ((A\otimes \openone)\sigma),\nonumber\\
&& \E (\mathbf{y}_B) = \Tr ((\openone\otimes A)\sigma).
\label{blqmsexp}
\end{eqnarray}
For our purposes the above requirements for expected values can be equivalently reformulated as conditions on joint probabilities. Indeed, any Hermitian operator with spectrum  $\{-1,1\}$ can be represented in terms of multidimensional projectors:
\begin{eqnarray}
A &=& 2P_A - \openone\nonumber\\
B &=& 2P_B - \openone.
\end{eqnarray}
Denoting by $p_{\alpha,\beta}$, $\alpha, \beta = \{-1,+1\}$, the joint probability:
\begin{equation}
p_{\alpha,\beta} \equiv \Pr [\mathbf y_A = \alpha \wedge \mathbf y_B = \beta],
\end{equation}
we can formulate the following set of conditions:
\begin{eqnarray}
&& p_{1,1} = \Tr ((P_A\otimes P_B)\sigma)\nonumber\\
&& p_{-1,1} = \Tr (((\openone-P_A)\otimes P_B)\sigma)\nonumber\\
&& p_{1,-1} = \Tr ((P_A\otimes (\openone-P_B)\sigma)\nonumber\\
&& p_{-1,-1} = 1 - p_{1,1} - p_{-1,1} - p_{1,-1}.
\label{blqmsprob}
\end{eqnarray}
Note that by \eqref{blqmsexp}-\eqref{blqmsprob}, we have
\begin{eqnarray}
&& \E (\mathbf{y}_A \mathbf{y}_B) = 1 - 2p_{-1,1} - 2p_{1,-1}\nonumber\\
&& \E (\mathbf{y}_A) = -1 + 2p_{1,1} + 2p_{1,-1}\nonumber\\
&& \E (\mathbf{y}_B) = -1 + 2p_{1,1} + 2p_{-1,1}.
\end{eqnarray}
Thus, the tuples $(\E (\mathbf{y}_A \mathbf{y}_B), \E (\mathbf{y}_A), \E (\mathbf{y}_B))$ and $(p_{1,1}, p_{-1,1}, p_{1,-1})$ can be uniquely determined from each other by linear transformation, and we have the following claim.
\begin{proposition}
For any given inputs $(A,\sigma)$ and $(B,\sigma)$, a $BLQMS$ protocol preserves the set of conditions (\ref{blqmsexp}) on expected values of outcomes in quantum measurement simulation if and only if it preserves the set of conditions (\ref{blqmsprob}) on joint probabilities.
\label{propo1}\qed
\end{proposition}

In the proof of the main result we make use of the lower bound for \emph{nondeterministic communication complexity} of the Deutsch-Jozsa problem~\cite{DJ92}, which in the  vector-based representation can be formulated as follows. Alice and Bob are given an $n$-bit vector $a$ and $b$ respectively, with each coordinate equal to $-1$ or $+1$: $a, b \in \{-1,1\}^n$, with $n$ even. A promise is made on the value of the scalar product of the inputs, which may take one of two values: $a\cdot b \in \{0, n\}$. The task of Alice and Bob is to decide which of the two cases holds, \emph{i.e.}, whether $a \cdot b = 0$ (equivalently, if $a$ and $b$ are equal on exactly $n/2$ coordinates) or $a\cdot b = n$ (equivalently, if $a$ and $b$ are equal on all $n$ coordinates). Formally, we can define the task function:
\begin{equation}
f(a,b) = [a\cdot b = n],
\end{equation}
where the square bracket denotes the logical value.

We recall that the \emph{nondeterministic communication complexity} of the function $f$ for accepting (rejecting), denoted $N^1(f)$ (respectively, $N^0(f)$), is the size of the smallest string which can be provided by an oracle to both Alice and Bob to convince them with certainty that $f(a, b) = 1$ (respectively, $f(a,b) = 0$). The following lower bound on $N^1(f)$ holds.

\begin{proposition}\label{pro:djn}
$N^1(f) \geq 0.007 n / (\log_2 n+3) - 1$.
\end{proposition}
\begin{proof}
By the classic theorem of Aho, Ullman, and Yannakakis~\cite{AUY83}, we have:
$$
D(f) \leq (N^0(f)+1)(N^1(f)+1),
$$
where $D(f)$ is the deterministic two-party communication complexity of the function $f$. It is known~\cite{BCW98} that for the considered function $f$, $D(f) \geq 0.007n$. Moreover, we trivially have $N^0(f) \leq \lceil \log_2 n \rceil + 1$, since for a rejecting input (with vectors $a$ and $b$ differing on exactly half of the coordinates), the oracle can communicate to Alice and Bob the pair $(i, \alpha)$, with $i\in \{1,\ldots,n\}$, $\alpha \in \{-1,1\}$, representing a coordinate $i$ such that $a_i = \alpha$ and $b_i = -\alpha$; the existence of such a coordinate can be verified by Alice and Bob for each of these two equalities, respectively. The claim follows.
\end{proof}


We proceed to show that any protocol which solves $BLQMS$ for dimension $n$ can be used as a black-box for designing an oracle verification protocol in the nondeterministic accepting Deutsch-Jozsa problem. More precisely, we will show that for any protocol $BLQMS$ which communicates $\Omega(n/\log_2 n)$ bits with small probability ($o(1/n)$ for all inputs), there would exist sets of inputs of dimension $n$ and corresponding initializations of the (shared) randomness of this protocol which could be used to solve Deutsch-Jozsa on all $n$-bit vectors with nondeterministic communication complexity $N^1(f) = o(n / \log_2 n)$, a contradiction with Proposition~\ref{pro:djn}.

\emph{Construction of the lower bound.} Using the equivalent characterization of the $BLQMS$ (Proposition~\ref{propo1}), for the purpose of the lower bound we define a restricted version $\overline{BLQMS}$ of the protocol, in which we only require that the first of the conditions  \eqref{blqmsprob} is preserved, namely that:
\begin{equation}
\Pr\ [\mathbf y_A = 1 \wedge \mathbf y_B = 1] = \Tr ((P_A\otimes P_B)\sigma).
\label{p11}
\end{equation}
We consider the special class of inputs by adding the promise that $\sigma$ is a maximally entangled $n$-dimensional state, whereas $P_A$ and $P_B$ are $1$-dimensional projectors, of the form $P_A = \frac{1}{n}|a\rangle\langle a|$ and $P_B = \frac{1}{n}|b\rangle\langle b|$, where $(a, b)$ is a valid pair of inputs to the vector version of the Deutsch-Jozsa problem, i.e., satisfying $a,b \in \{-1,+1\}^n$ and $a\cdot b\in\{0,n\}$. For the maximally entangled state, the constraint \eqref{p11} can be rewritten as:
\begin{equation}
\Pr\ [\mathbf y_A = 1 \wedge \mathbf y_B = 1]= \frac{1}{n^2} (a\cdot b) = \frac{1}{n} f(a,b),
\label{p11n}
\end{equation}
where the function $f$ corresponds to a solution to the Deutsch-Jozsa problem.

In the above, the probability distribution is governed by the choice of randomness: the shared randomness accessible to Alice and Bob, and the private randomness (``coin tosses'') used by Alice and Bob individually during the execution of their protocol. Without loss of generality, we can in fact assume that all randomness is shared, and that the random choices made during an execution of the protocol are parameterized by a random vector $\lambda \in X$, where $\lambda$ is sampled uniformly from some probabilistic space $X$ with measure $\mu$, independent of the inputs $a$ and $b$. The pairs $(a, \lambda)$ and $(b, \lambda)$, satisfying $a,b \in \{-1,+1\}^n$, $a\cdot b\in\{0,n\}$, and $\lambda\in X$, will now be treated as the inputs of our restricted version $\overline{BLQMS}$ of the $BLQMS$ problem. The deterministic protocol solving $\overline{BLQMS}$ must decide on a value $g(a,b,\lambda) = [y_A = 1 \wedge y_B = 1] \in \{0,1\},$ and it will be considered correct if it achieves the correct probability distribution of outputs over choices of $\lambda$:
$$
\E_\lambda\  g(a,b,\lambda) = \frac{1}{n} f(a,b),
$$
where we use the notation: $\E_\lambda\ \cdot \equiv \int_{\lambda\in X}\ \cdot\ d\mu.$ We make the technical assumption that the considered protocols are always such that the function $g$ is measurable; otherwise, the concept of expectation of outcomes is hard to define.

Fix a deterministic protocol solving $\overline{BLQMS}$. Let $T(a, b, \lambda)$ be the number of bits communicated by Alice and Bob in total during the execution of this protocol for inputs $(a, \lambda)$ and $(b, \lambda)$. We will now show that the the probability distribution of $T(a, b, \lambda)$ has a heavy tail over $\lambda$ for some input pairs, regardless of the considered protocol for $\overline{BLQMS}$.

We first show that the ``power'' of the shared randomness of $\overline{BLQMS}$ in the context of our lower-bound argument is limited, and corresponds, roughly speaking, to $O(\log n)$ bits of additional information received from the oracle. We start with a lemma which provides a convenient partition of the set of input vectors $\{-1,1\}^n$ for $\overline{BLQMS}$ into $O(n^2)$ subsets, given a bound on the tail of the communication complexity of the $\overline{BLQMS}$ protocol. Each of the subsets of the partition will be associated with a corresponding random vector $\lambda$ for which the output of the  $\overline{BLQMS}$ protocol is ``useful'' from the perspective of simulating Deutsch-Jozsa. Since the number of subsets in the partition is polynomial in $n$, given any input for Alice and Bob, an oracle for Deutsch-Jozsa will be able to prompt the right value of $\lambda$ to Alice and Bob using $O(\log n)$ bits of additional information. (We remark that eliminating shared randomness is a well-known concept in communication complexity~\cite{N91}, but our case concerns a different scenario of non-deterministic communication complexity.) 

\begin{lemma}\label{theLemma}
Assume that for some input size $n$ and some $M > 0$, the considered protocol for $\overline{BLQMS}$ on any pair of input vectors  will have a message complexity not smaller than $M$ with probability at most $\frac 1 {2n}$:
$$\forall_{a,b \in \{-1,+1\}^n} \mu (\{\lambda : \  T(a,b,\lambda) \geq  M \}) < \frac 1 {2n}.$$
Then there exists a partition of the set of $n$-dimensional input vectors: $\{-1,1\}^n = S_1 \cup \ldots \cup S_{2 n^2}$, such that each of the sets of inputs $S_j$ is associated with a fixed vector $\lambda_j \in X$ such that the following properties are fulfilled:
\begin{enumerate}
\item $g(a,a,\lambda_j) = 1, \quad \text{ for all } a\in S_j.$
\item $T(a,a,\lambda_j) < M, \quad \text{ for all } a\in S_j.$
\end{enumerate}
In other words, $\lambda_j \in X$ leads to an ``accepting'' outcome to $\overline{BLQMS}$ within less than $M$ communicated bits when Alice and Bob both receive the same input vector from $S_j$.
\end{lemma}
\begin{proof}
Let $A_0 = \{-1,1\}^n$ be the set of $n$-dimensional vectors, with $|A_0| = 2^n$. Consider any input to the $\overline{BLQMS}$ with $f(a,b)=1$, i.e., with $a=b \in A_0$. Let:
$$\Lambda(a) = \{\lambda \in X :  g(a,a,\lambda) = 1 \wedge T(a,a,\lambda) < M\}.$$
We have:
$$
\mu (\{\lambda \in X :  g(a,a,\lambda) = 1\}) = \E_\lambda\  g(a,a,\lambda) = \frac 1 n
$$
and by the assumption of the Lemma:
$$
\mu (\{\lambda \in X :  T(a,a,\lambda) \geq M \}) < \frac 1 {2n},
$$
hence $\mu(\Lambda(a)) \geq \frac 1 n - \frac 1 {2n} = \frac 1 {2n}$. Let $c_0 : X \to \{0,\ldots, 2^n\}$ be a function which assigns to each $\lambda$ the number of inputs $a\in A_0$ for which randomness $\lambda$ leads to an ``accepting'' outcome, formally, $c_1 = |S_1(\lambda)|$ with $S_1(\lambda) = \{a \in A_0 : \lambda \in \Lambda(a)\}$. Function $c_1$ is clearly measurable, and: $\E_\lambda c_1 = \sum_{a\in A_0} \mu(\Lambda(a)) = |A_0| / n = 2^n / (2n)$. By the pigeon-hole principle, it follows that there exists $\lambda_1 \in X$ with $c_1(\lambda_1) \geq  2^n / (2n)$ and so $|S_1(\lambda_1)| \geq 2^n/ (2n)$. Let $A_1$ be defined as the set of inputs for which randomness $\lambda_1$ does not lead to an accepting outcome, $A_1 = A_0 \setminus S_1(\lambda_1)$, with $|A_1| \leq |A_0| (1-1/(2n))$. We iterate the process, defining in general $A_{i} = A_{i-1} \setminus S_i(\lambda_i)$, where $S_i(\lambda_i) = \{a \in A_{i-1} : \lambda_i \in \Lambda(a)\}$ and $\lambda_i$ is chosen so that $|S_i(\lambda_i)| \geq |A_{i-1}| / (2n)$. It follows that $|A_{i}| \leq |A_{i-1}| (1-1/(2n))$, and so $|A_{i}| \leq |A_1| (1-1/(2n))^{i-1} < 2^n e^{-i/(2n)} = e^{n \ln 2 - i/(2n)}$. The process is completed when $|A_i| < 1$, i.e., $(2\ln 2) n  - i/n < 0$, or $i > (2\ln 2) n^2$ (where $2\ln 2 < 2$). Thus, the family of sets $S_1(\lambda_1), \ldots, S_{i}(\lambda_{i})$, for i = $2 n^2$, constitutes a partition of the set $A_0$ of $n$-dimensional vectors.
\end{proof}

We are now ready to provide our main technical theorem, showing that communication complexity distribution of any $\overline{BLQMS}$ protocol has a heavy tail.
\begin{theorem}\label{theClaim}
Fix any even value $n > 10^7$ and let $M(n) = 0.003 n/\log_2 n$.  For any protocol for $\overline{BLQMS}$, there exists a pair of input vectors $(a,b)$ of length $n$ for such that the protocol requires at least $M(n)$ communicated bits for this input pair with probability at least $\frac 1 {2n}$, taken over the shared randomness, i.e.:
$$\mu (\{\lambda : \  T(a,b,\lambda) \geq  M(n) \}) \geq \frac 1 {2n}.$$
\end{theorem}

\begin{proof}
The proof proceeds by contradiction. Suppose that the claim of the Theorem is not true for some value of $n$, and consider a protocol for $\overline{BLQMS}$ which violates the claim for this value of $n$. This means that the assumptions of Lemma~\ref{theLemma} hold for the considered protocol on inputs of dimension $n$ and $M\equiv M(n)$. To complete the proof, we will now show that the existence of such a protocol can be used to construct a nondeterministic accepting-outcome oracle and verifier for the Deutsch-Jozsa problem on vectors of size $n$ having nondeterministic-accepting communication complexity violating the lower bound on $N^1(f)$. Consider the following oracle-verifier pair for the Deutsch-Jozsa problem:

\paragraph{The oracle:} For a given accepting input $(a,b)$ with vectors of length $n$ satisfying the Deutsch-Jozsa promise (i.e., $a=b$), the oracle finds the set $S_j \subseteq A_0$ , which satisfies $a\in S_j$ in the partition given by Lemma~\ref{theLemma}, and the corresponding randomness $\lambda_j \in X$. It then performs a simulation of the execution of the $\overline{BLQMS}$ according to the considered protocol, achieving an accepting outcome ($g(a,a,\lambda_j) = 1$.) Finally, the oracle communicates the following string to both Alice and Bob:
\begin{enumerate}
\item The integer value $j \in \{1,\ldots,2n^2\}$.
\item A complete transcript of the messages communicated between Alice and Bob in the simulation of the $\overline{BLQMS}$ protocol for Alice's input $(a,\lambda_j)$ and Bob's input $(a,\lambda_j)$.
\end{enumerate}

\paragraph{The verifier:} Alice and Bob both receive the value of $j$ from the oracle, and use it to compute the corresponding randomness $\lambda_j \in X$. Then, Alice initializes a simulation of her part of the $\overline{BLQMS}$ protocol assuming her input as $(a,\lambda_j)$. Whenever the  $\overline{BLQMS}$ protocol would require her to receive a message from Bob, she reads the treats the corresponding part of the transcript received from the oracle as the message received from Bob. Whenever Alice sends a message in the protocol, she verifies if the message she sent corresponds to that written in the transcript provided by the oracle. If and only if there are no discrepancies and the outcome of Alice in the  $\overline{BLQMS}$ protocol is $1$, the Deutsch-Jozsa verifier returns in ``accepting'' state. Bob simulates the $\overline{BLQMS}$ protocol likewise for the input $(a,\lambda_j)$, trusting messages provided by the oracle for the communication from Alice to Bob, verifying all messages sent from Bob to Alice, and signaling ``accepting'' state if and only if there are no discrepancies and the outcome of Bob in the $\overline{BLQMS}$ protocol is $1$.

\bigskip
We now proceed to show the correctness of the oracle-verifier construction for the Deutsch-Jozsa problem. First, observe that for any $n$-bit Deutsch-Jozsa input $(a,b)$ satisfying $a=b$, the value of $\lambda_j$ prompted by the oracle satisfies $g(a,a,\lambda_j) = 1$, hence both Alice and Bob obtain an outcome of $1$ in the $\overline{BLQMS}$ protocol. The verification process proceeds successfully, and the verifier returns in accepting state for both Alice and Bob. Conversely, suppose that the Deutsch-Jozsa input $(a,b)$ satisfies $a\cdot b = 0$, and that the oracle provides some integer value $j'$ and some message transcript to Alice and Bob. Alice and Bob both compute $\lambda_{j'}$ and perform the simulation of the $\overline{BLQMS}$ protocol for respective inputs $(a,\lambda_{j'})$ and $(b,\lambda_{j'})$. If the transcript provided by the oracle corresponds precisely to a simulation of the considered $\overline{BLQMS}$ protocol for these inputs, then Alice and Bob do not detect any errors and compute their $\overline{BLQMS}$ outputs. Since $g(a,b,\lambda_{j'}) = \frac{1}{n} f(a,b) = 0$, at least one of the parties must return an output of $-1$ in $\overline{BLQMS}$, thus signaling a non-accepting outcome of the verification process.

Now, consider the length of the string communicated by the (correctly functioning) oracle to Alice and Bob for an accepting input $(a,b)$, with $a=b$. Encoding the value of $j$ requires $\log_2 (2n^2) = 2\log_2 n + 1$ bits. By Lemma~\ref{theLemma}, the total number of bits exchanged in the simulation of the $\overline{BLQMS}$ protocol is less than $M(n)$ bits, thus the transcript can be easily encoded using less than $2M(n)$ bits (including markers describing which parts of the message were communicated by Alice and which were communicated by Bob). It follows that the total length of the string provided by the oracle to the verifier is at most $2\log_2 n + 2M(n)$, and therefore this value is lower-bounded by the nondeterministic-accepting communication complexity of the Deutsch-Jozsa problem for inputs of length $n$:
$$2\log_2 n + 2M(n) \geq N^1(f) \geq 0.007 n / (\log_2 n+3) - 1,$$
hence:
$$
M(n) \geq 0.0035 n / (\log_2 n+3) - \log_2 n - 0.5.
$$
However, we have $M(n) = 0.003 n / \log_2 n$ by definition, and we obtain a contradiction for all $n>10^7$, which completes the proof.
\end{proof}

The lower bound on the tail of the number of communicated bits in the protocol, given by the above theorem, obviously also applies to the $BLQMS$ problem, of which $\overline{BLQMS}$ was a restricted case. Thus, we directly obtain lower bounds on the variance and higher moments of  the communication complexity of any protocol for $\overline{BLQMS}$.
\begin{theorem}[Lower bound on BLQMS]
For any protocol solving exactly the $BLQMS$ problem, let $T^{(k)}(n) = \max_{a,b\in\{-1,+1\}^n}\E_\lambda\ T^k(a, b, \lambda)$ be the $k$-th (non-central) moment of the random variable describing the number of communicated bits, taken over the probability distribution of $\lambda$, for a worst-case pair of inputs $a,b$ of length $n$. Then for any even $n>10^7$ the following lower bound holds:
\begin{equation*}
T^{(k)}(n) \geq \frac{1}{2n} M^k(n) = 0.5 \frac{(0.003 n)^{k-1}}{\log_2^k n}.
\vspace{-9mm}
\end{equation*}
\qed
\end{theorem}
In particular, it follows that the second non-central moment has an asymptotic growth rate of $T^{(2)}(n) = \Omega(n / \log^2 n)$, and $\limsup_{n\to+\infty} T^{(2)}{n} = +\infty$. In this way, we obtain the main result of our paper:
\begin{theorem}
If for some protocol for the $BLQMS$ problem the random variable describing the number of communicated bits has finite expectation for any input pair of arbitrary dimension, its variance must be unbounded.
\end{theorem}


\emph{Conclusions.} Using the methods of nondeterministic communication complexity, we have derived a lower bound on the communication complexity of the Binary Local Quantum Measurement Simulation, showing that any protocol solving this problem exactly can at most have a finite expected number of communicated bits, whereas all higher statistical moments of this quantity are necessarily  unbounded with respect to the input size (which is directly related to the dimension of the space of quantum states, on which the measurements are performed). We pose as an open question whether there exists such a protocol working with finite expected communication for an arbitrary quantum state.

\emph{Acknowledgements.} We would like to thank Caslav Brukner and Marek \.Zukowski for helpful discussions. MM is supported by the International PhD Project: ``Physics of future quantum based information technologies'' grant MPD/2009-3/4, the NCN grant No.\ 2012/05/E/ST2/02352 and by the project QUASAR.

\bibliography{SimQC}

\end{document}